\documentclass[conference]{IEEEtran}

\IEEEoverridecommandlockouts
\usepackage[utf8]{inputenc} 
\usepackage[T1]{fontenc}
\usepackage{url}
\usepackage{ifthen}
\usepackage{cite}
\usepackage[cmex10]{amsmath} 

\usepackage{amsthm}
\usepackage{bbm}
\usepackage{mathtools}
\usepackage{subcaption}
\usepackage{graphicx}
\usepackage{amsfonts}
\usepackage{xcolor}

\usepackage{enumitem}

\theoremstyle{plain}
\newtheorem{theorem}{Theorem}
\newtheorem{lemma}{Lemma}
\newtheorem{corollary}{Corollary}
\newtheorem{proposition}{Proposition}
\newtheorem{protocol}{Protocol}
\theoremstyle{definition}
\newtheorem{rem}{Remark}

\newcommand{\vect}[1]{\boldsymbol{\mathbf{#1}}}

\begin{document}
\title{Byzantine-Resilient Distributed Computation via Task Replication and Local Computations} 


\author{%
}

\author{%
  \IEEEauthorblockN{Aayush Rajesh \thanks{AR's work was done while at IIT Bombay. NK's work was supported by a SERB MATRICS grant. The work of VP was supported by DAE, Govt. of India, under project no. RTI4001.}}
  \IEEEauthorblockA{Stanford University, USA \\ aayushr@stanford.edu}
  \and 
  \IEEEauthorblockN{Nikhil Karamchandani}
  \IEEEauthorblockA{IIT Bombay, India \\ nikhilk@ee.iitb.ac.in}
  \and 
  \IEEEauthorblockN{Vinod M. Prabhakaran}
  \IEEEauthorblockA{TIFR, India \\ vinodmp@tifr.res.in}
}

\maketitle


\begin{abstract}
   We study a distributed computation problem in the presence of Byzantine workers where a central node wishes to solve a task that is divided into independent sub-tasks, each of which needs to be solved correctly. The distributed computation is achieved by allocating the sub-task computation across workers with replication, as well as solving a small number of sub-tasks locally, which we wish to minimize due to it being expensive. For a general balanced job allocation, we propose a protocol that successfully solves for all sub-tasks using an optimal number of local computations under no communication constraints. Closed-form performance results are presented for cyclic allocations. Furthermore, we propose a modification to this protocol to improve communication efficiency without compromising on the amount of local computation.
\end{abstract}

\section{Introduction}
The massive scale of modern datasets has necessitated the use of distributed computing architectures, such as MapReduce \cite{dean2008mapreduce} and Apache Spark \cite{zaharia2010spark}, for various data analytics and machine learning pipelines. Carefully designed distributed processing of jobs can lead to significant speedups in various ubiquitous tasks such as matrix multiplication \cite{benson2015framework,bulucc2012parallel}, training of deep neural networks \cite{dean2012large}, optimization \cite{boyd2011distributed}, and Markov Chain Monte Carlo (MCMC) \cite{neiswanger2014asymptotically}. 

We consider a master-worker distributed computing architecture, where the central node is assigned a `large' task which it splits into several `small' sub-tasks, which can be distributed among a collection of worker nodes. The worker nodes are expected to perform their assigned sub-tasks and then communicate the results to the central node, which aggregates them appropriately to complete the original task. There are several challenges in such a system; one which has garnered the most attention is that the overall latency is determined by the computation speed of the slowest worker, which can significantly hamper the end-to-end system performance \cite{dean2013tail,li2014communication,yadwadkar2016multi}. One solution to this problem of \textit{straggler nodes} is to introduce redundancy into the task allocation by replicating the same task at multiple workers. There has been a line of systems works which have developed this idea \cite{dean2008mapreduce,ananthanarayanan2013effective,chen2013improving}, as well as several theoretical studies on characterizing the optimal redundancy-latency tradeoff \cite{joshi2017efficient,gardner2015reducing,wang2014efficient,behrouzi2021efficient}. Beyond replication, there is also a large body of work employing erasure coding for straggler mitigation \cite{dutta2016short,lee2017speeding,yu2020straggler,karakus2017straggler,li2018near,pmlr-v70-tandon17a}.

Distributed computing systems are also more susceptible to errors in the reported results from some workers, which could be due to software/hardware malfunctions or the presence of malicious agents \cite{kim2014flipping, castro1999practical,brun2011smart}. In this work, we consider the \textit{Byzantine adversaries} model \cite{lamport2019byzantine}, where we assume some subset (of at most a certain specified size) of workers is malicious and their response to their allotted tasks can be completely arbitrary. As before, the goal of the central node is to recover the correct results for all the sub-tasks irrespective of the actions of the malicious workers. To assist in this, we assume that the central node can (occasionally) query an oracle to perform \textit{local computation} of a chosen task and obtain its true result. The main goal of our work is to characterize the optimal tradeoff between the replication factor and the number of local computations needed, for any given bound on the number of possible malicious workers. To achieve this, we first establish a converse that provides a lower bound on the number of local computations needed in the worst case, and then provide a sub-task allocation and decoding strategy whose performance matches the converse. 

The impact of Byzantine adversaries has been studied in a wide variety of distributed computing settings, ranging from matrix multiplication and polynomial computations \cite{solanki2019non,hong2023byzantine,keshtkarjahromi2019secure,yu2019lagrange} to gradient aggregation for model training \cite{konstantinidis2023detection,chen2018draco,rajput2019detox}, and optimization \cite{data2020data, kuwaranancharoen2020byzantine}. Our work differs from these in that we allow for local computations at the central node, motivated by \cite{10206794, hofmeister2025byzantine,jain2024interactive}, which first studied the tradeoff with respect to local computations for the problem of \textit{gradient coding} \cite{pmlr-v70-tandon17a} under Byzantine adversaries. There, the central node is interested in computing the gradient with respect to a loss function over a large dataset, and each worker node is assigned a subset of the data over which it needs to find the gradient. We can think of the computation of the gradient over each individual datapoint as a sub-task, and then the overall task is to find the sum of the results of individual tasks (full gradient). On the other hand, in our setting, the tasks are arbitrary, and the central node needs to recover all the individual results correctly.



The rest of the paper is organized as follows. Section~\ref{Sec:Setting} describes the problem setting and Section~\ref{Sec:LB} establishes a lower bound on the required number of local computations for any `balanced' task allocation scheme. Section~\ref{Sec:FullC} provides the design and analysis for our proposed scheme, while Section~\ref{Sec:LessC} discusses a variant that requires much lower communication between the workers and the central node. 

\section{Problem Setting}
\label{Sec:Setting}
We consider the problem of solving a task $T$, whose result can be obtained by dividing it into multiple independent sub-tasks $\{T_i\}$ and then aggregating the outputs of all the sub-tasks, i.e., $T = f(T_1, \ldots, T_p)$. We make no assumptions on the structure of this function $f$. Thus, we treat the problem of solving the task $T$ as equivalent to distributing the computation and solving each of the sub-tasks.

A central node tasked with performing this computation distributes the $p$ sub-tasks among $n$ workers (possibly with some redundancy), of which at most $s$ workers are malicious. This distribution of sub-tasks can be represented in the form of a job allocation matrix $\vect{A} \in \{0,1\}^{p \times n}$ where $\vect{A}_{ij}$ is the indicator of whether the $i$\textsuperscript{th} sub-task is assigned to the $j$\textsuperscript{th} worker. Let $ \mathcal{D}_i \subseteq \{1, \ldots, p\}$ be the set of sub-tasks allocated to the $i$\textsuperscript{th} worker, whom we denote by $W_i$. 

Each worker computes the results corresponding to their allotted sub-tasks. Thereafter, the central node communicates with the workers to obtain their responses to the allotted sub-tasks. Each worker $W_i$ maintains a collection of responses $\Tilde{\vect{T}}_j \ \forall j \in \mathcal{D}_i$. For an honest worker, $\Tilde{\vect{T}}_j$ equals the true results corresponding to the assigned sub-tasks. However, the malicious workers may report incorrect results for some or all of their allotted sub-tasks.

To assist the central node, it is given access to an oracle that it may query to perform a local computation and obtain the true result of a chosen sub-task. A protocol employed by the central node may proceed over several rounds, with each round involving communication from the central node to the workers requesting particular responses, communication from the workers to the central node containing their responses based on the values they maintain for each allotted sub-task, and possible local computations performed by the central node. The responses asked by the central node may either be the value of a particular sub-task or commitments to the value of a sub-task reported by some other worker.

We are primarily concerned with the performance of a  protocol against the following three figures of merit:

\begin{itemize}
    \item \textbf{Local Computations: }The total number of sub-tasks computed locally at the central node, denoted by $c$.
    \item \textbf{Replication Factor: }The average number of workers to which each sub-task is assigned, denoted by
    $$
    \rho = \frac{\sum_{i \in [p], j \in [n]} \vect{A}_{ij}}{p}.
    $$
    \item \textbf{Communication Overhead: }The maximum number of symbols from the underlying alphabet transmitted from the workers to the central node, represented by $\kappa$.
\end{itemize}

Since the replication factor $\rho$ is a proxy for redundancy, we would like to have job allocations that achieve low values of $\rho$. For a fixed choice of job allocation, and hence replication, we would like to have an associated protocol capable of obtaining the correct results for each of the $p$ sub-tasks, and thus the correct result of the overall task, while requiring as few local computations as possible, irrespective of the actions of malicious workers.

In practice, we wish to have balanced job allocations, i.e., where all workers get the same number of sub-tasks. Our work focuses on general balanced job allocations, with closed-form bounds presented for balanced cyclic allocations. As an example, a balanced cyclic allocation distributing $p = n$ sub-tasks among $n$ workers with replication factor $\rho$ is depicted in Figure \ref{fig: balanced cyclic allocation}.

\begin{figure}[htbp]
\centering
$$
\vect{A} = 
\begin{bmatrix}
    \smash{\overbrace{\begin{matrix} 1 & 1 & \cdots & 1 & 1 \end{matrix}}^{\rho}} & \begin{matrix} 0 & 0 & \cdots & 0 & 0 \end{matrix}\\
    \begin{matrix} 0 & 1 & \cdots & 1 & 1 \\ \vdots & \vdots & \ddots & \vdots & \vdots \\ 1 & 1 & \cdots & 1 & 0 \end{matrix} & \begin{matrix} 1 & 0 & \cdots & 0 & 0 \\ \vdots & \vdots & \ddots & \vdots & \vdots \\ 0 & 0 & \cdots & 0 & 1 \end{matrix}
\end{bmatrix}_{n \times n}
$$
\caption{Example of a balanced cyclic allocation}
\label{fig: balanced cyclic allocation}
\end{figure}

It is clear that if a sub-task is allotted to at most $s$ workers, i.e., $\rho \le s$, in the case that all these workers are malicious and may collude to report an incorrect result, we would not be able to identify if there has been a corruption in the sub-task responses. The only way to correctly identify the correctness of every sub-task result would be to perform local computations on all of them in the worst case. Furthermore, if a sub-task is allotted to at least $2s+1$ workers, i.e., $\rho \geq 2s+1$, a simple majority over responses would suffice since the sub-task would have at most $s$ malicious workers, and hence at least $s+1$ honest ones. Therefore, it is meaningful to look at job allocations with a replication factor $\rho = s+u$ where $u$ is an integer taking values in $[1,s]$. In balanced job allocations, this guarantees that each sub-task is allotted to at least $u$ honest workers.

\section{Local Computation Lower Bound}
\label{Sec:LB}
We first establish a lower bound on the local computation requirement of any scheme by constructing an adversarial attack, which is a modified version of the symmetrization attack discussed in \cite{10206794}. We constrain the attack by targeting a block sub-matrix of the job allocation matrix with the original symmetrization attack.

\begin{proposition}
    Let the largest $k \times uk$ sub-matrix (not necessarily contiguous) of 1's in the (balanced) job allocation matrix with replication $\rho = s+u$ have $k = k^*$. An adversary that inflicts exactly $u$ corruptions in each row and exactly $1$ corruption in each column of this sub-matrix will force $c \geq \min\{k^*, \lfloor \frac{s}{u} \rfloor\}$ local computations.
    \label{prop: constrained symm}
\end{proposition}

Note that we require taking the minimum between $k^*$ and $\lfloor \frac{s}{u} \rfloor$, because such an attack with $s$ malicious workers can only target at most $\lfloor \frac{s}{u} \rfloor$ sub-tasks. Such an attack would then be feasible targeting $\min\{k^*, \lfloor \frac{s}{u} \rfloor\}$ sub-tasks and having $u \min\{k^*, \lfloor \frac{s}{u} \rfloor\} \leq s$ workers behave maliciously in the largest $k \times uk$ sub-matrix. 

The attack described in the proposition is the symmetrization attack from \cite{10206794} confined to a sub-matrix of 1's in the job allocation matrix. Each corrupted sub-task $i$ has $u$ malicious workers incorrectly reporting $\vect{T}_i ''$ and $s$ workers correctly reporting $\vect{T}_i '$. This gives rise to two indistinguishable scenarios from the central node's view in each corrupted sub-task: either the $u$ workers are malicious and the $s$ workers are a mixture of malicious and honest workers reporting the correct value, or the $u$ workers are honest and the $s$ workers are malicious. By performing a local computation on the $i$\textsuperscript{th} sub-task, $u$ malicious workers are eliminated, leaving $s-u$ remaining malicious workers. Now, in the remaining corrupted sub-tasks, we have $u$ workers reporting one value ($\vect{T}_j ''$) and $s-u$ workers reporting another ($\vect{T}_j '$), and the problem persists. Therefore, we are forced to perform local computations on each of the corrupted sub-tasks, i.e., at least $\min\{k^*, \lfloor \frac{s}{u} \rfloor\}$ such local computations. This argument is made more formal in Appendix \ref{app: constrained symm attack formal}. An example of this attack is depicted in Figure \ref{fig:Constrained Symm}, which tabulates responses in the transpose of the job allocation matrix, with the worker responses corresponding to the largest $k \times uk$ (in this case $2\times 4$) sub-matrix of 1's in the job allocation matrix highlighted in blue. 

\begin{rem}
    If we consider regimes where the number of sub-tasks $p$ is much larger than the number of workers $n$, we may eventually find a $k \times uk$ sub-matrix of 1's where $k \geq \lfloor \frac{s}{u} \rfloor$. Then, the above result implies that an attack exists, such that any protocol will require at least $\lfloor \frac{s}{u} \rfloor$ local computations.
    \label{rem: large sample regime}
\end{rem}

\begin{rem}
    The $\lfloor \frac{s}{u} \rfloor$ local computation bound can be trivially achieved. For any sub-task where fewer than $u$ workers agree on a value, those workers can be declared as malicious, as any true value will have at least $u$ honest workers agreeing with it. We can then sequentially perform local computations on all sub-tasks with multiple reported values, with each value having at least $u$ agreeing workers. Each local computation will eliminate at least $u$ malicious workers, requiring at most $\lfloor \frac{s}{u} \rfloor$ of them.
    \label{rem: trivial bound protocol}
\end{rem}

\begin{figure}[htbp]
    \centering
    \includegraphics[width=0.6\linewidth]{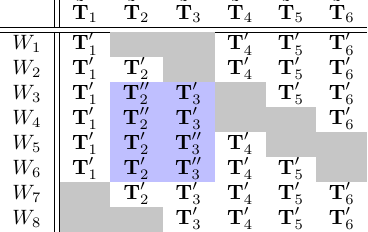}
    \caption{Example of attack for $s = 4, u = 2$ } 
    \label{fig:Constrained Symm}
\end{figure}

While the lower bound on local computation in Proposition \ref{prop: constrained symm} depends on $k^*$, a closed-form expression for $k^*$ may not always be viable for arbitrary job allocations. However, there is such an expression for the specific case of a balanced cyclic allocation.

\begin{proposition}
    For a balanced cyclic allocation matrix with $p = \lambda n$ for $\lambda \in \mathbbm{Z}^+$, and $\rho = s+u$, the largest $k$ such that there exists a $k \times uk$ sub-matrix of 1's in the allocation matrix is given by,
    \begin{equation}
    k_{\lambda}^* = \left\lfloor \frac{\rho + 1}{u + \frac{1}{\lambda}} \right\rfloor
    \label{eq: sub matrix size}
    \end{equation}
    Furthermore, for $\lambda n < p < (\lambda + 1)n$, the corresponding largest value of $k$ is bounded by 
    $$k_{\lambda}^* \leq k^* \leq \min \left\{k_{\lambda}^* + p - \lambda n, k_{\lambda + 1}^* \right\}$$
    \label{prop: sub matrix size}
\end{proposition}
This result is described in Appendix \ref{app: Submatrix of cyclic matrix}. Combining Propositions \ref{prop: constrained symm} and  \ref{prop: sub matrix size}, we have the following corollary.
 
\begin{corollary}
    For a balanced cyclic allocation with $\lambda n \leq p < (\lambda + 1)n$ where $\lambda \in \mathbbm{Z}^+$, and $\rho = s+u$, there exists an adversarial attack such that any protocol requires $c \geq \min \left \{k_{\lambda}^*, \lfloor\frac{s}{u}\rfloor \right \}$, where $k_{\lambda}^*$ is as defined in \eqref{eq: sub matrix size}.
    \label{cor: cyclic lower bound}
\end{corollary}

\section{Full Communication}
\label{Sec:FullC}
In this section, we propose a protocol to solve the problem under a full communication relaxation, i.e., all workers report the values of all their allotted sub-tasks to the central node. Here, our focus is on the local computation performance.

\subsection{Full Communication Protocol}

We describe our protocol using a ``logic-based'' protocol as a sub-routine, which is discussed in Section \ref{ssec: Logic subroutine}.

\begin{protocol}{Full Communication Protocol}
    \begin{enumerate}
        \item In any sub-task, if strictly less than $u$ workers agree on a value, mark these workers as malicious and eliminate them.
        \item Rank the corrupted sub-tasks based on the decreasing order of the number of disagreements among workers in the sub-tasks. For example, a sub-task where the largest agreeing set is $s - 4$ has more disagreement than one where the largest such set is $s - 1$.
        \item We then greedily perform a local computation on the sub-task with the most disagreement (if there are multiple, then we choose randomly) and eliminate responses of those workers found to be malicious. Recompute the disagreements in the remaining sub-tasks.
        \item Repeat steps 1-3 until we reach a stage where each corrupted sub-task has $\leq u$ disagreeing workers. From this stage, we use the ``logic-based protocol'' (see Appendix \ref{app: logic protocol description}) to solve the remaining sub-tasks and hence the task.
    \end{enumerate}
    \label{full comm protocol}
\end{protocol}

Note that we need to recompute the disagreement in step 3 since we have eliminated some worker responses after the greedy local computation. We show that our protocol performs optimally for a balanced cyclic allocation, in the sense that it achieves the lower bound derived in Corollary \ref{cor: cyclic lower bound}.

\begin{theorem}
    For a balanced cyclic allocation with $p = \lambda n$ and $\rho = s+u$, Protocol \ref{full comm protocol} requires $c \leq \min \{ k_{\lambda}^*, \lfloor \frac{s}{u} \rfloor \}$ local computations to solve the sub-tasks, where $k_{\lambda}^*$ is as in \eqref{eq: sub matrix size}. 
    \label{th: full comm performance}
\end{theorem}

Before proving Theorem \ref{th: full comm performance} in Section \ref{ssec: Proof of Theorem}, we describe the subroutine used in step 4 of Protocol \ref{full comm protocol} and prove some of its properties.

\subsection{Logic-Based Scheme} \label{ssec: Logic subroutine}

We show that under a special kind of attack, balanced job allocation matrices allow one to use logical arguments to eliminate certain malicious workers without even requiring local computations. This is exploited by our ``logic-based'' protocol, which is described in detail in Appendix \ref{app: logic protocol description}. The correctness of the actions of the protocol is shown in Lemma \ref{lem: u corrupt correctness}, while its performance is shown in Lemma \ref{lem: u corrupt bound}. Note that Protocol \ref{full comm protocol} makes no assumption regarding the attack, and only calls the ``logic-based'' protocol as a subroutine when the view after greedy local computations matches this special case.

\begin{lemma}
For $\rho = s+u$ and a balanced job allocation, if each corrupted sub-task has exactly $u$ workers with responses disagreeing with the rest, the protocol in Appendix \ref{app: logic protocol description} correctly eliminates malicious workers and solves for sub-task values. 
\label{lem: u corrupt correctness}
\end{lemma}

The proof is presented in Appendix \ref{app: correctness proof} and is via a case analysis of the ``logic-based'' protocol in Appendix \ref{app: logic protocol description}. The first loop in the protocol goes over pairs of corrupted sub-tasks and looks for specific cases in their responses where it is easy to identify malicious workers using no/fewer local computations, based on a proof by contradiction. As an example, if we assume the $u$ disagreeing workers in a sub-task to be honest, the remaining $s$ will be malicious. However, one of these $u$ workers may disagree with a new worker in another sub-task, making that one malicious as well, giving rise to more than $s$ malicious workers. This contradiction implies that the set of workers assumed to be honest is malicious. The second loop performs local computations on the remaining corrupted sub-tasks where no easier resolution is possible. 

\begin{lemma}
    The maximum number of local computations required by the protocol in Appendix \ref{app: logic protocol description} is when the adversary inflicts a constrained symmetrization attack on the largest $k \times uk$ sub-matrix of 1's.
    \label{lem: u corrupt bound}
\end{lemma}

\begin{proof}

In each case of the first loop, we either eliminate malicious workers without using local computations or eliminate $>u$ malicious workers in each local computation. In the case of a constrained symmetrization attack, each disagreeing worker disagrees in only one corrupted sub-task, and each disagreeing worker has been allotted all corrupted sub-tasks. Therefore, it can be seen that none of the cases in the first loop of the protocol occur. Instead, we perform local computations on all $k$ corrupted sub-tasks in the second loop, with each local computation eliminating exactly $u$ malicious workers. Thus this attack is the worst-case performance of the logic-based protocol.
    
\end{proof}

Since the worst-case performance of the logic-based protocol corresponds to a constrained symmetrization attack, we have the following corollary.

\begin{corollary}
    For a balanced cyclic allocation with $p = \lambda n$ and $\rho = s + u$, when each sub-task has exactly $u$ disagreeing workers, the number of local computations $c$ performed by the ``logic-based'' protocol satisfies $c \leq \min \left \{k_{\lambda}^*, \lfloor\frac{s}{u}\rfloor \right \}$, where $k_{\lambda}^*$ is as defined in \eqref{eq: sub matrix size}.
    \label{cyclic u corruption}
\end{corollary}

\subsection{Proof of Theorem \ref{th: full comm performance}} \label{ssec: Proof of Theorem}

When the logic-based protocol is called in Protocol \ref{full comm protocol}, each sub-task has $\leq u$ disagreeing workers. Any set of $< u$ disagreeing workers in a sub-task can be treated as malicious, as each sub-task has at least $u$ honest workers. Therefore, we only have sub-tasks with exactly $u$ disagreeing workers, which is a valid input to the logic-based protocol.

\begin{proof}

Suppose the greedy phase of the protocol expends $l$ local computations. Since we only target sub-tasks that have more than $u$ disagreements, each local computation reveals $\geq (u+1)$ malicious workers. Therefore, when the greedy phase stops with corrupted sub-tasks having at most $u$ disagreeing workers, we have $\leq s - (u+1)l$ malicious workers left. 

Firstly, any sub-task with an agreeing group of size $< u$ or $> s - (u+1)l$ can be solved trivially since there are at least $u$ honest workers in each sub-task and at most $s - (u+1)l$ malicious workers. In the worst case, there are no such sub-tasks and all we are left with are sub-tasks with groups of size $u$ and $(s-(u+1)l-i) \forall \, i \geq 0$. 

From this point onward, our logic-based argument takes over. From Lemma \ref{lem: u corrupt bound}, the worst case number of corruptions corresponds to the largest $k \times uk$ sub-matrix of 1's. The size of the largest $k \times uk$ sub-matrix in this case, is bounded by the corresponding size when all the sub-tasks have groups of size $u$ and $s-(u+1)l$ since we obtain the former by deleting rows from the latter, an action that can only reduce the size of the largest sub-matrix. This corresponds to the case when $s' = s-(u+1)l$ and $\rho' = s' + u$. Therefore, from  Corollary \ref{cyclic u corruption}, we require less than $ \min \left \{ \left\lfloor \frac{\rho' + 1}{u + \frac{1}{\lambda}} \right\rfloor, \lfloor\frac{s'}{u}\rfloor \right \}$ local computations in this stage. Putting it all together,

\begin{align*}
    c &\leq l + \min \left \{ \left\lfloor \frac{s-(u+1)l + u + 1}{u + \frac{1}{\lambda}} \right\rfloor, \left\lfloor\frac{s-(u+1)l}{u}\right\rfloor \right \} \\
    &= l +  \min \left \{ \left\lfloor \frac{s + u + 1}{u + \frac{1}{\lambda}} -  \frac{l(u+1)}{u + \frac{1}{\lambda}}\right\rfloor, \left\lfloor\frac{s}{u} - \frac{l(u+1)}{u}\right\rfloor \right\} \\
    &\leq l + \min \left \{\left\lfloor \frac{s + u + 1}{u + \frac{1}{\lambda}} -  l \right\rfloor, \left\lfloor\frac{s}{u} - l\right\rfloor \right \}  \quad \quad \quad (\lambda \geq 1)\\
    &= \min \left \{\left\lfloor \frac{s + u + 1}{u + \frac{1}{\lambda}}\right\rfloor, \left\lfloor\frac{s}{u}\right\rfloor \right \}
\end{align*}
    
\end{proof}

\subsection{Computation and Communication Performance}

We will take the ratio of our local computation upper bound to the trivial upper bound of $\lfloor \frac{s}{u} \rfloor$ and characterize its improvement. In Figure \ref{fig: Full Comm Comp} we plot the variation of this ratio for Protocol \ref{full comm protocol} with the number of sub-tasks for increasing values of the parameter $u$, which increases redundancy. Solid points indicate achievable local computation upper bounds for our protocol (corresponding to $p$ as a multiple of $n$), while for intermediate values of $p$, a loose upper bound is plotted based on the bound in Proposition \ref{prop: sub matrix size}.

\begin{figure}[htbp]
    \centering
    \includegraphics[width=0.85\linewidth]{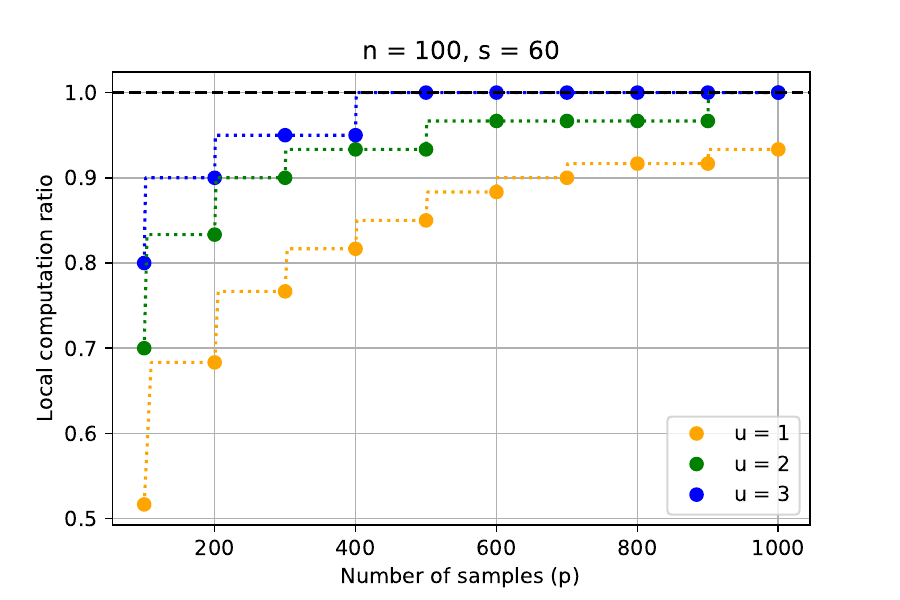}
    \caption{Local computation ratio ($\downarrow$) of performance of Protocol \ref{full comm protocol} to trivial performance of $\lfloor \frac{s}{u} \rfloor$ for three examples}
    \label{fig: Full Comm Comp}
\end{figure}

The nature of the plots suggests that under increasingly higher redundancy, the performance of our protocol has lower gains with respect to the trivial computation bound. Furthermore, as the number of sub-tasks increases, the computation performance of the protocol begins to converge to the trivial bound. This echoes an observation made in Remark \ref{rem: large sample regime}: while staggered job allocations reduce the potency of an adversary (by reducing the size of the largest sub-matrix), in regimes where $p$ is large compared to $n$, the adversary becomes more powerful as larger block sub-matrices are now possible.

Notice that Protocol \ref{full comm protocol} is a one-shot protocol, i.e., it receives all allotted sub-task results from all workers and uses this data to learn sub-task values with/without local computations. Therefore, it has a constant communication requirement of $\kappa = p \rho$ independent of the choice of adversarial attack. As the protocol uses full communication, it can be expected that this is a high communication overhead.

\section{Communication-Efficiency}
\label{Sec:LessC}
In this section, we propose modifications to Protocol \ref{full comm protocol} in the form of commitments aimed at reducing the higher communication overhead while retaining the local computation properties of the protocol.

\subsection{Commitment-Based Protocol}

Rather than receiving responses from all workers for all allotted sub-tasks, we nominate a single (possibly unique) worker for each sub-task, and that worker alone sends the result for the particular sub-task. From the rest of the workers allotted to each sub-task, we require just a single-bit commitment (yes/no) to the value put forward by the nominated worker. By doing this, we greatly reduce our communication requirement from that of the full communication case, as most of the communication now corresponds to single-bit values. This comes at the cost of an increased number of rounds.

\begin{protocol}{Commitment-Based Protocol}
    \begin{enumerate}
        \item Choose a single worker in each sub-task to act as the representative and obtain the designated sub-task result from each representative. Obtain single-bit commitments to these values from all other workers.
        \item Run Protocol \ref{full comm protocol} on this commitment table to identify malicious workers, and when possible, learn sub-task. Eliminate the identified malicious workers and retain only those sub-tasks whose results have not yet been learned.
        \item Repeat steps 1-2 until all sub-task results are learned.
    \end{enumerate}
    \label{commitment protocol}
\end{protocol}

\begin{lemma}
    For a given job allocation matrix and number of malicious workers, Protocol \ref{commitment protocol} will not require any more local computations than those required by Protocol \ref{full comm protocol}.
    \label{lem: commitment vs full comm}
\end{lemma}

The proof is presented in Appendix \ref{app: commitment vs full comm proof}.

\subsection{Communication Performance}

In the execution of Protocol \ref{commitment protocol}, the set of unsolved sub-tasks in each round is a subset of that in the previous round. Therefore, in an attack where the set of unsolved sub-tasks remains the same over each round, the total number of rounds (and hence communication overhead) will be maximum.

\begin{proposition}
    The communication overhead (ignoring single-bit terms) of Protocol \ref{commitment protocol} satisfies $\kappa \leq p + \left(\lfloor \frac{s}{u} \rfloor - 1 \right) u$
    \label{prop: worst communication}
\end{proposition}

In the worst case, each of the representatives chosen in corrupted sub-tasks in each round are malicious. In the first round, we have $u$ corrupted sub-tasks, with the $u$ chosen representatives agreeing in each sub-task and disagreeing with the remaining workers. Note that we require exactly $u$ workers agreeing, as more than $u$ such workers would cause a local computation on that sub-task, hence solving it and reducing the set of unsolved sub-tasks. An instance of the described attack in a round is depicted in Figure \ref{fig:commitment worst case}. In this example, $W_3$ and $W_4$ behave maliciously, agreeing with each other on sub-tasks 2 and 3 (denoted by single-bit commitment 1).

Following the execution of Protocol \ref{full comm protocol}, these $u$ workers are marked as malicious without local computations. However, since they were the representatives for their respective sub-tasks, these sub-tasks remain unsolved going into the next round. The analysis is carried out similarly in the following rounds.

Such an attack can only be carried out for a total of $\lfloor \frac{s}{u} \rfloor$ rounds. In the first round, there are $p$ sub-task values being reported (one for each of the $p$ sub-tasks). In the following rounds, there are $u$ sub-task values being reported (one for each of the $u$ unsolved sub-tasks). Hence we have the stated communication overhead bound.

\begin{figure}[htbp]
    \centering
    \includegraphics[width=0.6\linewidth]{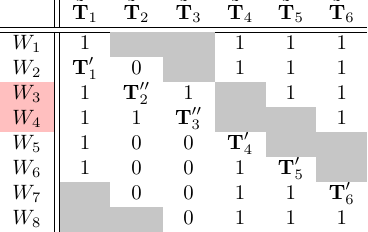}
    \caption{Example of worst case in a round for $s = 4, u = 2$}
    \label{fig:commitment worst case}
\end{figure}

\clearpage

\bibliographystyle{IEEEtran}
\bibliography{references}

\clearpage

\onecolumn

\appendix

\subsection{Logic-Based Protocol Description} \label{app: logic protocol description}

\begin{figure*}[htbp]
    \centering
    \includegraphics[width = 0.88\textwidth]{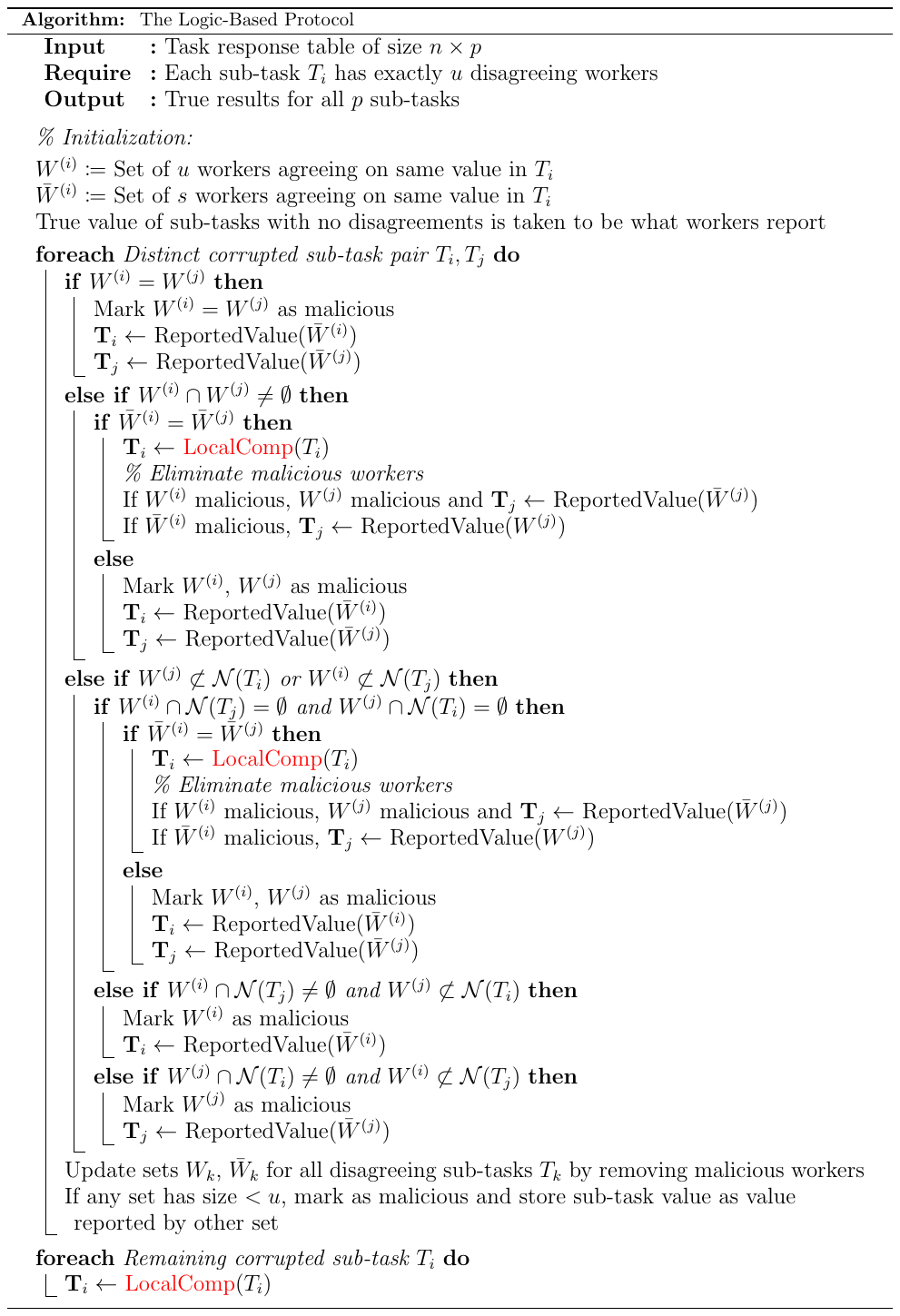} 
    \label{fig:logic algo}
\end{figure*}

\twocolumn

\subsection{Constrained Symmetrization Attack}\label{app: constrained symm attack formal}

Suppose that we have a protocol which handles the attack described in Proposition \ref{prop: constrained symm} with $c$ local computations, where $c$ satisfies $c < \min\{k^*, \lfloor \frac{s}{u} \rfloor\}$ local computations. Then, there is a sub-task $\Tilde{\vect{T}}_i$ in the largest $k \times uk$ sub-matrix that has not been locally computed, where $u$ workers report one value and $s$ workers report another. The $c$ local computations performed so far have revealed $cu$ malicious workers, and these $cu$ malicious workers are a strict subset of the $s$ workers in $\Tilde{\vect{T}}_i$. Therefore, among workers not found to be malicious, sub-task $\Tilde{\vect{T}}_i$ has $u$ workers reporting one value and $s - cu$ workers reporting another, but there are also $s-cu$ malicious workers left. 

This gives rise to two indistinguishable scenarios where either set may be malicious. Therefore, an extra local computation is required on this sub-task, contradicting the fact that the protocol requires $c$ local computations. Hence, we cannot have a protocol requiring strictly fewer than $\min\{k^*, \lfloor \frac{s}{u} \rfloor\}$ local computations.

\subsection{Sub-matrix size of balanced cyclic matrices}\label{app: Submatrix of cyclic matrix}

Let us first tackle the first expression of Proposition \ref{prop: sub matrix size}. The balanced cyclic job allocation with $p = \lambda n$ consists of $\lambda$ repeating blocks of a cyclic square matrix. Furthermore, in a cyclic square matrix, any $t$ consecutive data samples have $\rho - (t - 1)$ workers in common. Since each row in the cyclic square matrix repeats $\lambda$ times in the job allocation matrix, we have that $\lambda t$ data samples have $\rho - (t-1)$ workers in common. For these $\lambda t$ data samples to be part of a $k \times uk$ sub-matrix of 1's,

\begin{align*}
    u\lambda t &\leq \rho - (t - 1) \\
    \implies (u\lambda + 1)t & \leq \rho + 1 \\
    \implies \lambda t &\leq \frac{\rho + 1}{u + \frac{1}{\lambda}}
\end{align*}

Therefore, the largest $k$ such that the sub-matrix exists is given by $k_{\lambda}^*$ mentioned in the proposition. 
Now we move on to the second expression of the proposition. Note that when $\lambda n < p < (\lambda + 1)n$, the largest sub-matrix size is bounded between those when $p = \lambda n$ and $p = (\lambda + 1)n$ since the addition of rows (increasing $p$) can only increase the size of the largest sub-matrix. Furthermore, with every row addition (increment of $p$ by 1), the largest sub-matrix size can only increase by at most 1. Putting these two facts together, we get a bound on the value of $k^*$ for a balanced cyclic job allocation matrix.

\subsection{Proof of Lemma \ref{lem: u corrupt correctness}} \label{app: correctness proof}

We can treat the job allocation matrix $\vect{A}$ as an adjacency matrix to describe a bipartite graph, with an edge existing between a worker and a sub-task if the sub-task was allotted to the worker. For balanced job allocations, this describes a semi-regular bipartite graph $\mathcal{G} = (W, D, E)$, where $W$ is the set of workers, $D$ is the set of sub-tasks, and $E$ is the collection of edges $(W_i, T_j)$ such that sub-task $T_j$ is allotted to worker $W_i$. 

The neighborhood of a set of nodes $A$ (denoted by $\mathcal{N}(A)$) is defined as the set of nodes that share an edge with a node in $A$. Using this definition, we will represent the sub-tasks allotted to worker $W_i$ as $\mathcal{N}(W_i)$, and the workers allotted to sub-task $T_j$ as $\mathcal{N}(T_j)$. Therefore, $|\mathcal{N}(T_j)| = \rho$ and $|\mathcal{N}(W_i)| = p\rho/n$.

\begin{proof}

Let us denote $W^{(i)} \vcentcolon= \left\{W_1^{(i)}, W_2^{(i)}, \ldots, W_u^{(i)}\right\}$ to be the set of $u$ disagreeing workers for corrupted sub-task $T_i$. We shall also denote by $\bar{W}^{(i)} \vcentcolon =\{\bar{W}_j^{(i)}\}_{j = 1}^{s}$, the $s$ agreeing workers, i.e. workers in the set $\mathcal{N}(T_i) \setminus W^{(i)}$. Now, we shall consider the set of disagreeing workers $W^{(i)}$ and $W^{(j)}$ for any two distinct corrupted sub-tasks $T_i$ and $T_j$ respectively. By distinct, we mean that it should not be the case that both $W^{(i)} = W^{(j)}$ and $\mathcal{N}(T_i) \setminus W^{(i)} = \mathcal{N}(T_j) \setminus W^{(j)}$, since if they were equal then we only need to consider one of these sub-tasks, as identifying malicious workers in one of them also does so in the other. Then, we have the following three (disjoint) possibilities:
\begin{enumerate}[leftmargin=*]
    \item \underline{$W^{(i)} = W^{(j)}$}:\\
    Let us assume the set of $u$ workers $W^{(i)}$ to be honest. Therefore, we have that the set of $s$ workers $\mathcal{N}(T_i) \setminus W^{(i)}$ is malicious. Similarly, since $W^{(j)} = W^{(i)}$, we have that the set of $s$ workers $\mathcal{N}(T_j) \setminus W^{(j)}$ is also malicious.\\
    However, these two sets are each of size $s$ and are not equal. Therefore, we have $>s$ potential malicious workers, which is a contradiction. Therefore, our assumption is false, and the set $W^{(i)} = W^{(j)}$ is malicious, which has been concluded without a local computation. 

    \item \underline{$W^{(i)} \cap W^{(j)} \neq \emptyset$}:\\
    In this scenario, we have two possible cases:
    \begin{enumerate}
        \item $\mathcal{N}(T_i) \setminus W^{(i)} = \mathcal{N}(T_j) \setminus W^{(j)}$: \\
        In this case, there is no way to distinguish the two sub-tasks, so we must perform a local computation on any one of them. Suppose we perform a local computation of $T_i$.\\
        If this reveals that $\mathcal{N}(T_i) \setminus W^{(i)}$ is malicious, then we have found all of our $s$ malicious workers using a single local computation, and these are also the $s$ agreeing workers in $T_j$. Therefore, we learn that the sets $W^{(i)}, W^{(j)}$ are honest workers.\\
        Alternatively, if we learn that the set $W^{(i)}$ is malicious, then we can eliminate the workers in this set, which also eliminates some workers in $W^{(j)}$, giving us $W^{(j)'} = W^{(j)} \setminus W^{(i)}$ (because $W^{(i)} \cap W^{(j)} \neq \emptyset$). Now, in $T_j$, we have one agreeing group with $s$ workers and one with $< u$ workers ($W^{(j)'}$). But the number of remaining malicious workers is $<s$, and hence the group of $s$ workers must be reporting the true value. Therefore, $W^{(j)'}$ is also malicious, and we have found $|W^{(i)} \cup W^{(j)}| > u$ malicious workers using a single local computation.
        \item $\mathcal{N}(T_i) \setminus W^{(i)} \neq \mathcal{N}(T_j) \setminus W^{(j)}$:\\
        Let us assume the set of $u$ workers $W^{(i)}$, to be honest. Therefore, the set $\mathcal{N}(T_i) \setminus W^{(i)}$ are our $s$ malicious workers.\\
        Suppose $\exists W_k^{(i)} \in W^{(i)}$ such that $W_k^{(i)} \in \mathcal{N}(T_j) \setminus W^{(j)}$. Then $W_k^{(i)}$ is, by assumption, honest and disagrees on $T_j$ with another honest (by assumption) worker $W_l^{(j)} \in W^{(i)} \cap W^{(j)}$, which is a contradiction. Therefore, our assumption is false, and the set $W^{(i)}$ (and hence $W^{(j)}$) is malicious, which has been concluded without a local computation. \\
        Instead, suppose we are unable to find a worker such that the above holds. Then $\exists \bar{W}_k^{(j)} \in (\mathcal{N}(T_j) \setminus W^{(j)})\setminus(\mathcal{N}(T_i) \setminus W^{(i)})$ which has not been assumed as honest or malicious because it neither belongs to $W^{(i)}$ nor $\mathcal{N}(T_i) \setminus W^{(i)}$. However, $\bar{W}_k^{(j)}$ disagrees with some honest (by assumption) worker $W_l^{(j)} \in W^{(i)} \cap W^{(j)}$, which means that $\bar{W}_k^{(j)}$ is also malicious, giving rise to $>s$ malicious workers. This is a contradiction, meaning our assumption is false, and the set $W^{(i)}$ (and hence $W^{(j)}$) is malicious, which has been concluded without a local computation.
    \end{enumerate}

    \item \underline{$W^{(i)} \cap W^{(j)} = \emptyset$}:\\
    Suppose not all disagreeing workers have been allotted all corrupted sub-tasks. Then, we can find $T_i, T_j$ such that $W^{(j)} \not\subset \mathcal{N}(T_i)$. If $T_i, T_j$ are such that $W^{(i)} \cap \mathcal{N}(T_j) = \emptyset$ and $W^{(j)} \cap \mathcal{N}(T_i) \neq \emptyset$, then without loss of generality, we can take $T_i$ to be $T_j$ and vice versa. Then, we have the following two possible cases:

    \begin{enumerate}
        \item $W^{(i)} \cap \mathcal{N}(T_j) \neq \emptyset$:\\
        Let us assume the set of $u$ workers $W^{(i)}$, to be honest. Therefore, the set $\mathcal{N}(T_i) \setminus W^{(i)}$ are our $s$ malicious workers.\\
        In $T_j$, we have some honest (by assumption) $W_k^{(i)} \in W^{(i)} \cap \mathcal{N}(T_j)$ such that $W_k^{(i)} \in \mathcal{N}(T_j) \setminus W^{(j)}$  because $W^{(i)} \cap W^{(j)} = \emptyset$. We also have (by assumption) some $W_l^{(j)} \in W^{(j)}$ such that $W_l^{(j)} \not\in \mathcal{N}(T_i)$ and therefore has not been assumed as honest or malicious. However, $W_k^{(i)}$ and $W_l^{(j)}$ disagree, giving rise to $>s$ malicious workers, which is a contradiction. Therefore, our assumption is false, and the set $W^{(i)}$ is malicious, which has been concluded without local computation.  
        \item $W^{(i)} \cap \mathcal{N}(T_j) = \emptyset$:\\
        Since we have $W^{(i)} \cap \mathcal{N}(T_j) = \emptyset$ and $W^{(j)} \cap \mathcal{N}(T_i) = \emptyset$, we have the following two sub-cases:
        \begin{enumerate}[label=(\roman*)]
            \item $\mathcal{N}(T_i) \setminus W^{(i)} = \mathcal{N}(T_j) \setminus W^{(j)}$:\\
            Just as in Case 2(a), we cannot distinguish between the two sub-tasks, so we must perform a local computation on any one of them. Suppose we perform a local computation on $T_i$.\\
            If this reveals that $\mathcal{N}(T_i) \setminus W^{(i)}$ is malicious, then we have found all of our $s$ malicious workers using a single local computation, and these are also the $s$ agreeing workers in $T_j$. Therefore, we learn that the sets $W^{(i)}, W^{(j)}$ are honest workers.\\
            Alternatively, if we learn that the set $W^{(i)}$ is malicious, we can eliminate these workers, which does not change the number of responses in $T_j$. Now, in $T_j$, we have one agreeing group with $s$ workers and one with $u$ workers, but the number of remaining malicious workers is $<s$. Therefore, the group of $s$ agreeing workers is reporting the true value, and the set $W^{(j)}$ is malicious. Thus, we have found $|W^{(i)} \cup W^{(j)}| = 2u$ malicious workers using a single local computation. 
            \item $\mathcal{N}(T_i) \setminus W^{(i)} \neq \mathcal{N}(T_j) \setminus W^{(j)}$:\\
            In this case, we can directly infer that both the sets $W^{(i)}$ and $W^{(j)}$ are malicious. Suppose we assume $W^{(i)}$ to be honest and thus $\mathcal{N}(T_i) \setminus W^{(i)}$ to be our $s$ malicious workers. \\
            Then, $\exists \bar{W}_k^{(j)} \in (\mathcal{N}(T_j) \setminus W^{(j)})\setminus(\mathcal{N}(T_i) \setminus W^{(i)}$ which has not been assumed as honest or malicious because it neither belongs to $W^{(i)}$ nor $\mathcal{N}(T_i) \setminus W^{(i)}$. However, $\bar{W}_k^{(j)}$ disagrees with workers in the set $W^{(j)}$, which have not been assumed as honest or malicious. This gives rise to $> s$ malicious workers, which is a contradiction. Therefore, our assumption is false, and the set $W^{(i)}$ is malicious. A symmetric argument also gives us that the set $W^{(j)}$ is malicious.
        \end{enumerate}
    \end{enumerate}

\end{enumerate}
    
\end{proof}

To visualize these scenarios, an example for each of the cases in the proof is provided in Figure \ref{fig: proof cases}. In the case of assumptions being made to prove a contradiction, workers who are assumed to be honest are marked in green, and those who are malicious are marked in red.

In Case 1, we are assuming that $W_4, W_5$ are honest, giving that $W_2, W_3, W_6, W_7$ in $\Tilde{\vect{T}}_2$ and $W_8$ in $\Tilde{\vect{T}}_3$ are malicious. In Case 2(b), we are assuming that $W_4, W_5$ are honest, giving that $W_3, W_6, W_7, W_8$ in $\Tilde{\vect{T}}_3$ and $W_9$ in $\Tilde{\vect{T}}_5$ are malicious. In Case 3(a), we are assuming that $W_3, W_4$ are honest, giving that $W_2, W_5, W_6, W_7$ in $\Tilde{\vect{T}}_2$ and $W_8$ in $\Tilde{\vect{T}}_3$ are malicious. In Case 3(b)(ii), we are assuming that $W_3, W_4$ are honest in $\Tilde{\vect{T}}_2$, giving rise to more malicious workers in $\Tilde{\vect{T}}_5$ as $W_5, W_6, W_7$ disagree with $W_8, W_9$.

\begin{figure}[!htbp]
    \centering
    \begin{subfigure}{0.87\linewidth}
        \centering
        \includegraphics[width = 0.9\linewidth]{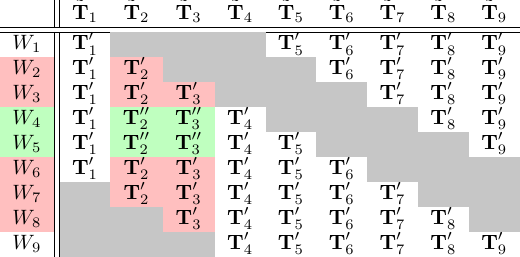}
        \caption{Case 1}
    \end{subfigure}
    \begin{subfigure}{0.87\linewidth}
        \centering
        \includegraphics[width = 0.9\linewidth]{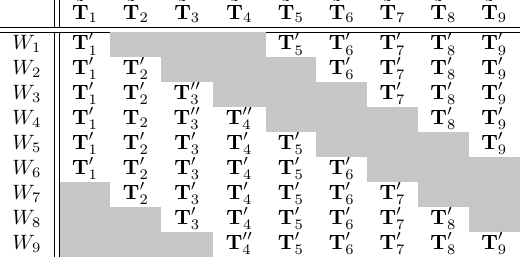}
        \caption{Case 2(a)}
    \end{subfigure}
    \begin{subfigure}{0.87\linewidth}
        \centering
        \includegraphics[width = 0.9\linewidth]{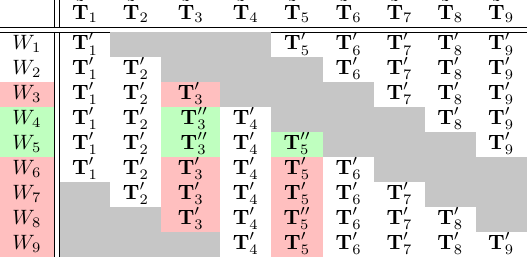}
        \caption{Case 2(b)}
    \end{subfigure}
    \begin{subfigure}{0.87\linewidth}
        \centering
        \includegraphics[width = 0.9\linewidth]{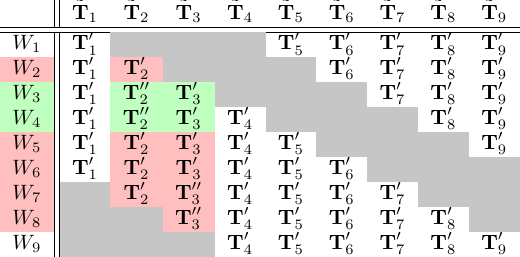}
        \caption{Case 3(a)}
    \end{subfigure}
    \begin{subfigure}{0.87\linewidth}
        \centering
        \includegraphics[width = 0.9\linewidth]{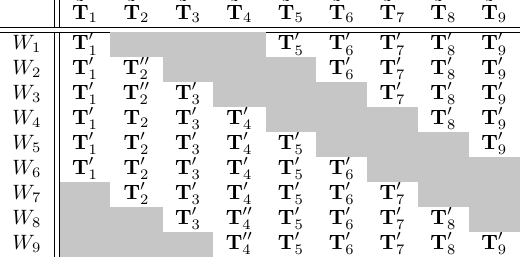}
        \caption{Case 3(b)(i)}
    \end{subfigure}
    \begin{subfigure}{0.87\linewidth}
        \centering
        \includegraphics[width = 0.9\linewidth]{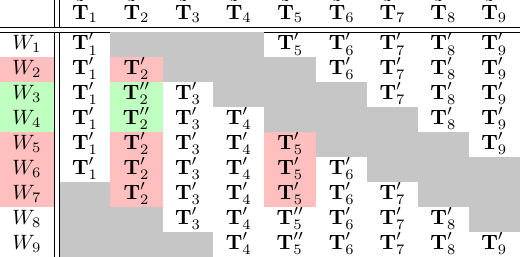}
        \caption{Case 3(b)(ii)}
    \end{subfigure}
    \caption{Example of proof cases for $n = p = 9, s = 4, u = 2$}
    \label{fig: proof cases}
\end{figure}

\subsection{Proof of Lemma \ref{lem: commitment vs full comm}} \label{app: commitment vs full comm proof}

\begin{proof}

Let $k^*$ be the largest $k$ such that there is a $k \times uk$ sub-matrix of 1's in the job allocation matrix. As seen earlier, this is also the local computation bound for Protocol \ref{full comm protocol}. Suppose in an execution of Protocol \ref{commitment protocol}, we have performed $k^*$ local computations (say, on sub-tasks $\vect{T}_1, \ldots, \vect{T}_{k^*}$) and there is still a sub-task left (say, $\vect{T}'$) with its value unknown. If we term each iteration of steps 1 and 2 in Protocol \ref{commitment protocol} as a ``round'', then in the first round, $\vect{T}_1, \ldots, \vect{T}_{k^*}, \vect{T}'$ were corrupted sub-tasks. 

Among these $k^* + 1$ sub-tasks, there must exist at least one pair of sub-tasks such that their responses satisfy any one of the cases listed in the proof of Lemma \ref{lem: u corrupt correctness}. This is because none of those cases hold if and only if the sub-tasks and their responses resemble a symmetrization attack on the largest $k \times uk$ sub-matrix. However, since we have more sub-tasks than the size of the largest possible sub-matrix ($k^*$), at least one case will hold. 

Furthermore, in every one of these cases, the protocol performs a local computation on at most one sub-task. Since we assumed that the sub-tasks $\vect{T}_1, \ldots, \vect{T}_{k^*}$ were locally computed, the pair where such a case occurs must be of the form $(\vect{T}', \vect{T}'')$ where $\vect{T}'' \in \{\vect{T}_1, \ldots, \vect{T}_{k^*}\} $. 

Lastly, notice that when a local computation is done on $\vect{T}''$ in Protocol \ref{full comm protocol}, it also solves $\vect{T}'$ by reducing the number of malicious workers in the sub-task to $< u$. Therefore, $\vect{T}'$ can be solved without a local computation, ensuring that Protocol \ref{commitment protocol} requires no more than the $k^*$ local computations.

\end{proof}

\end{document}